\documentclass{article}

\usepackage[a4paper, margin=1.2in]{geometry}
\usepackage{graphicx} 
\usepackage[utf8]{inputenc} 
\usepackage[T1]{fontenc}    
\usepackage{hyperref}       
\usepackage{url}            
\usepackage{booktabs}       
\usepackage{amsfonts}       
\usepackage{nicefrac}       
\usepackage{microtype}      
\usepackage{xcolor}         
\usepackage{subcaption}
\usepackage{nicematrix}
\usepackage{algorithm}
\usepackage{algorithmic}

\usepackage{float}
\usepackage{amssymb}
\usepackage{amsmath}
\usepackage{dsfont}
\usepackage{amsthm}
\usepackage{bm}
\usepackage{mathtools}
\usepackage{xcolor}
\usepackage{xspace}
\usepackage[toc]{appendix}
\newtheorem{assumption}{Assumption}
\newtheorem{definition}{Definition}
\newtheorem{lemma}{Lemma}
\newtheorem{theorem}{Theorem}

\usepackage{paralist}
\usepackage{bbm}

\title{Dynamic Regret in Time-varying MDPs with Intermittent Information}

\author{
Negin Musavi$^{1}$ and Melkior Ornik$^{1}$\\
{\small $^{1}$Department of Aerospace Engineering and the Coordinated Science Laboratory,}\\
{\small University of Illinois Urbana--Champaign, Champaign, IL 61820, USA} \\
{\tt\small nmusavi2@illinois.edu, mornik@illinois.edu}
}

\date{}

\begin{document}

\maketitle

\section*{Abstract}\label{sec:abs}

We study sequential decision-making in time-varying Markov decision processes (TVMDPs) under limited update rates, where the decision-maker observes the system and updates its model only intermittently. Such settings arise in applications with sensing, communication, or computational constraints that preclude continuous adaptation. Our goal is to understand how the performance of an agent, which learns and plans using receding-horizon control under these information constraints, degrades as a function of the update rate. We propose a skip-update learning and planning framework that combines likelihood-based estimation of time-varying transition kernels with finite-horizon planning and executes policies between updates using stale information. We analyze its performance via dynamic regret relative to an oracle policy with full knowledge of the dynamics and continuous observations. Our main result establishes a dynamic regret bound that explicitly quantifies the impact of intermittent updates, decomposing regret into contributions from update times and skip intervals and revealing its dependence on temporal variation, estimation uncertainty, and the duration of intervals without updates. In particular, the dominant contribution from skip intervals admits a linear dependence on the interval length and the rate of temporal variation, while its effect is mitigated by mixing-induced contraction.
\section{Introduction}\label{sec:intro}

We consider sequential decision-making problems in which the system dynamics are time-varying and the available information and computational resources are limited in rate. In many real-world applications, sensing, communication, and computation constraints prevent continuous access to state information and continuous model updates. As a result, the decision-maker receives observations only intermittently and must operate in between using stale or incomplete information.

\medskip
\noindent
Such constraints arise naturally in networked systems with bandwidth limitations, robotic platforms with restricted sensing or energy budgets, and large-scale data-driven systems where frequent updates are computationally impractical. In these settings, both information acquisition and computation are resource-constrained, forcing the decision-maker to operate for extended periods without new data or policy updates.

\medskip
\noindent
To study this problem, we model sequential decision-making in time-varying Markov decision processes (TVMDPs), where the transition dynamics evolve over time and the decision-maker receives observations and performs updates only at selected times. Between these updates, the agent must act using previously acquired information and previously computed policies. A central question in this setting is:

{\centering
\emph{How does performance degrade when the decision-maker is forced to operate for extended periods without new information or updates?}
\par}

\medskip
\noindent
During these intervals, two sources of error arise: (i) \emph{model mismatch}, due to temporal variation in the dynamics, and (ii) \emph{state and policy mismatch}, due to acting based on outdated information and outdated computations. Understanding how these errors accumulate over time, and how they depend on the rate of information acquisition and computation, is the main focus of this paper.

\medskip
\noindent
Reinforcement learning and decision-making in Markov decision processes (MDPs) have been extensively studied in stationary environments, where the transition dynamics are fixed over time. A large body of work focuses on learning optimal policies under unknown but stationary dynamics, with theoretical guarantees based on regret or sample complexity \cite{kearns2002near, brafman2002r, kolter2009near, ouyang2017learning, ji2023regret}. These works provide algorithms that balance exploration and exploitation, but rely fundamentally on the assumption that the environment does not change over time.

\medskip
\noindent
To address more realistic scenarios, several works extend this framework to non-stationary or time-varying environments. Early formulations introduce time-dependent MDPs and continuous-time models to capture temporal variations in dynamics and rewards \cite{boyan2000exact}. More recent work studies online learning and control in time-varying MDPs, providing dynamic regret guarantees that scale with the variation of the environment \cite{li2019online}. Similarly, adaptive regret and related metrics have been proposed to capture performance in changing dynamical systems \cite{gradu2023adaptive}. While these works account for temporal variation, they typically assume continuous access to observations and the ability to update policies at every time step.

\medskip
\noindent
In parallel, a growing literature considers learning and planning in explicitly time-varying environments. Time-varying MDP formulations and solution methods have been developed for applications such as robotics and navigation under spatiotemporal disturbances \cite{liu2018solution}. Learning-based approaches incorporate temporal variation through maximum likelihood estimation with drift constraints or uncertainty quantification \cite{ornik2021learning}, as well as extensions to partially observable settings \cite{puthumanaillam2024weathering}. In robotics and control, time-varying dynamics are also addressed using model predictive control, meta-reinforcement learning, and robust optimization techniques \cite{duckworth2021time, jiang2021attention, zouitine2024time}. 

Theoretical performance guarantees for model predictive schemes in nonstationary MDPs have also been established, including dynamic regret bounds under the assumption that future transition kernels are known or can be accurately predicted \cite{zhang2024predictive}. However, these approaches generally rely on frequent updates or continuous access to state information.

\medskip
\noindent
More recently, there has been interest in understanding the role of update frequency and computational constraints in reinforcement learning. For example, \cite{lee2024pausing} shows that continuously updating policies is not always optimal in non-stationary environments, and that introducing pauses in updates can improve performance. However, such works do not explicitly characterize how performance degrades as a function of the information or update rate.

\medskip
\noindent
\paragraph{Our Contribution.}
In contrast to existing literature, we study decision-making in time-varying MDPs under a limited rate of observations and updates. We explicitly model both information and computation constraints and characterize how performance degrades when the decision-maker operates for extended periods without new data or recomputation.

\medskip
\noindent
To this end, we propose a skip-update learning and decision-making framework for TVMDPs with limited observation and update rates. The key idea is to perform model estimation and policy updates only at selected times, and to reuse the resulting policy between updates. At each update time, the decision-maker constructs an estimate of the time-varying transition kernels using constrained maximum likelihood estimation that incorporates known drift bounds \cite{ornik2021learning}. It then solves a finite-horizon planning problem based on the estimated model and applies the resulting policy until the next update.

\medskip
\noindent
This leads to a piecewise-constant decision-making strategy that isolates how errors accumulate during intervals without updates and how this accumulation interacts with temporal variation in the dynamics.

\medskip
\noindent
From a theoretical perspective, we analyze the performance of the proposed algorithm through a dynamic regret criterion, comparing it to an oracle policy with full knowledge of the time-varying dynamics and continuous updates. Our analysis decomposes the regret into contributions arising from temporal variation, estimation and planning errors, finite-horizon effects, and intervals without new information and updates.

\medskip
\noindent
The resulting regret bound explicitly characterizes how the limited rate of information and updates affects performance in time-varying environments. In particular, it quantifies how errors accumulate during skip intervals and provides insight into the trade-off between resource constraints and decision-making performance.

\medskip
\noindent
The remainder of the paper is organized as follows. In Section~\ref{sec:prob}, we introduce the problem setting and the proposed learning and planning framework. Section~\ref{sec:regret} presents the dynamic regret analysis and the main theoretical results. Additional technical details and proofs are provided in the appendix.

\paragraph{Notation.}
We denote by $\mathbb{N}$ the set of nonnegative integers and by $\mathbb{R}$ the set of real numbers. 
\noindent
For a finite set $\mathcal{X}$, let $\Delta(\mathcal{X})$ denote the probability simplex over $\mathcal{X}$. For $\mu \in \Delta(\mathcal{X})$ and $f : \mathcal{X} \to \mathbb{R}$, define $\mathbb{E}_{x \sim \mu}[f(x)] := \sum_{x \in \mathcal{X}} \mu(x) f(x).$

\noindent
A probability transition kernel on $\mathcal{S}$ is a mapping $P : \mathcal{S} \to \Delta(\mathcal{S})$, where $P(\cdot \mid s)$ denotes the distribution of the next state given the current state $s$. A controlled transition kernel is a mapping $P : \mathcal{S} \times \mathcal{A} \to \Delta(\mathcal{S})$, where $P(\cdot \mid s,a)$ denotes the distribution of the next state given state $s$ and action $a$. We write $P_t(\cdot \mid s,a)$ for the transition kernel at time $t$. For a controlled transition kernel $P$ and $f : \mathcal{S} \to \mathbb{R}$, define $(Pf)(s,a) := \mathbb{E}_{s' \sim P(\cdot \mid s,a)}[f(s')].$

\noindent
For integers $k < k'$ and an action sequence $(a_k,\dots,a_{k'})$, let $P^{a_k \cdots a_{k'}}_{k:k'}(\cdot \mid s)$ denote the distribution of the state at time $k'$ starting from $s$ at time $k$. Correspondingly, $(P^{a_k \cdots a_{k'}}_{k:k'} f)(s) := \mathbb{E}[f(s_{k'})]$ under this evolution.

\noindent
For a policy $\pi : \mathcal{S} \to \Delta(\mathcal{A})$, the induced transition kernel is $P^\pi(\cdot \mid s) := \sum_{a \in \mathcal{A}} \pi(a \mid s) P(\cdot \mid s,a).$ A (time-varying) policy sequence is a collection $\{\pi_t\}_{t=0}^{T-1}$ with $\pi_t : \mathcal{S} \to \Delta(\mathcal{A})$.

\noindent
For $f : \mathcal{X} \to \mathbb{R}$, define the span seminorm $\mathrm{sp}(f) := \max_{x \in \mathcal{X}} f(x) - \min_{x \in \mathcal{X}} f(x),$ and for $\mu,\nu \in \Delta(\mathcal{X})$, the total variation distance $\|\mu - \nu\|_{\mathrm{tv}} := \tfrac{1}{2} \sum_{x \in \mathcal{X}} |\mu(x) - \nu(x)|.$

\noindent
We will repeatedly use the standard bound: for any controlled transition kernels $P,\hat{P}$ and any function $f : \mathcal{S} \to \mathbb{R}$, $\mathrm{sp} \big( Pf - \hat{P}f \big)
    \le
    \max_{(s,a)} \|P(\cdot \mid s,a) - \hat{P}(\cdot \mid s,a)\|_{\mathrm{tv}} \, \mathrm{sp}(f).$
\section{Problem Setting}\label{sec:prob}

We consider a finite-horizon time-varying Markov decision process (TVMDP)
\begin{align*}
    \mathcal{M} = \bigl(\mathcal{S}, \mathcal{A}, T, \{P_t\}_{t=0}^{T-1}, \{r_t\}_{t=0}^{T-1}\bigr),
\end{align*}    
where $\mathcal{S}$ and $\mathcal{A}$ denote finite state and action spaces, respectively, and $T \in \mathbb{N}$ is the horizon length. At each time step $t \in \{0,\ldots,T-1\}$, the transition kernel $P_t : \mathcal{S} \times \mathcal{A} \to \Delta(\mathcal{S})$ specifies a probability distribution $P_t(\cdot \mid s,a)$ over next states given the current state--action pair $(s,a)$, and the reward function $r_t : \mathcal{S} \times \mathcal{A} \to \mathbb{R}$ assigns an instantaneous reward $r_t(s,a)$. We assume that the reward function is known at any time $t$ , whereas the transition kernels $\{P_t\}_{t=0}^{T-1}$ are unknown and vary over time. We further assume that the dynamics evolve gradually and satisfy a bounded drift condition, known to the agent, for any $s,s' \in \mathcal{S}$, $a \in \mathcal{A}$, and $t \ge 0$:
\begin{equation}\label{eq:drift}
    \bigl| P_{t+1}(s' \mid s,a) - P_t(s' \mid s,a) \bigr|
    \le \varepsilon_t,
\end{equation}
where $\{\varepsilon_t\}_{t \ge 0}$ encodes prior knowledge on the rate of temporal variation of the dynamics. 

\medskip
\noindent
We assume that the agent does not receive new information at every time step. Instead, information is revealed only at a subset of time indices, which we refer to as update times:
\[
    \mathcal{T}_{\mathrm{upd}}
    =
    \{\tau_0,\tau_1,\ldots,\tau_{N_T}\},
    \quad
    0 \le \tau_0 < \tau_1 < \cdots < \tau_{N_T} < T.
\]
We assume $\tau_0 = 0$, so that the initial state is available to the agent. 
We define the set of skip times as
\[
    \mathcal{T}_{\mathrm{skip}}
    :=
    \{0,\ldots,T-1\} \setminus \mathcal{T}_{\mathrm{upd}}.
\]

\medskip
\noindent
At each update time $t \in \mathcal{T}_{\mathrm{upd}}$, the agent observes the current state $s_t$ and, after applying an action $a_t$, also observes the subsequent state $s_{t+1}$. Thus, update times correspond to time steps at which new information becomes available to the agent. The information collected at update times is used by the agent to refine its decision-making strategy. At skip times, the agent does not receive new information and instead operates based on the information available from the most recent update time. At every time step, regardless of whether an update is performed, the agent selects an action, receives the instantaneous reward $r_t(s_t, a_t)$ and the system evolves according to the underlying transition kernel.

\medskip
\noindent
For any time $t \geq 0$, define the index of most recent update time and the set of update times up to time $t$ as
\[
    k(t) := \max\{k : \tau_k \le t\},\quad \text{and} \quad 
    \mathcal{T}^t_{\mathrm{upd}}
    :=
    \{\tau_0,\tau_1,\ldots,\tau_{k(t)}\},
\]
respectively. The goal of the agent then is to learn a sequence of policies $\{\pi_t\}_{t=0}^{T-1}$ ($\pi_t: \mathcal{S} \to \mathcal{A}$) based on the available information, so as to maximize the expected cumulative reward, i.e.,
\begin{equation}
    \{\pi_t^\star\}_{t=0}^{T-1}
    = \operatorname*{arg\,max}_{\{\pi_t\}_{t=0}^{T-1}}
    \mathbb{E}\!\left[
        \sum_{t=0}^{T-1} r_t\big(s_t, \pi_t(s_t)\big)
    \right],
    \label{eq:objective}
\end{equation}
where the expectation is taken with respect to the sequence of transition kernels $\{P_t\}_{t=0}^{T-1}$. The corresponding optimal value of this objective is denoted by
\begin{equation}
    J_T^\star (s)
    = \max_{\{\pi_t\}_{t=0}^{T-1}}
    \mathbb{E}\!\left[
        \sum_{t=0}^{T-1} r_t\big(s_t, \pi_t(s_t)\big) \mid s_{0} = s
    \right],
    \quad \forall s \in \mathcal{S}.
    \label{eq:opt-benchmark}
\end{equation}

\medskip
\noindent
The optimal value function of the TVMDP $\mathcal{M}$ at time $t$ is defined as
\begin{align*}
    V_t^\star(s)
    :=
    \max_{\{\pi_k\}_{k=t}^{T-1}}
    \mathbb{E}\!\left[
        \sum_{k=t}^{T-1}
        r_k\bigl(s_k, \pi_k(s_k)\bigr)
        \mid
        s_t = s
    \right],
    \quad \forall s \in \mathcal{S},
\end{align*}
with terminal condition $V_T^\star \equiv 0$. The optimal value functions satisfy the Bellman recursion
\[
    V_t^\star(s)
    =
    \max_{a \in \mathcal{A}}
    \Big\{
        r_t(s,a)
        +
        \mathbb{E}_{s' \sim P_t(\cdot \mid s,a)}[V_{t+1}^\star(s')]
    \Big\},
    \qquad t = 0,\dots,T-1.
\]

\medskip
\noindent
The value function $V_t^\star$ characterizes the optimal performance that can be achieved when the transition kernels $\{P_t\}_{t=0}^{T-1}$ are known and the system state is fully observed at every time step, in which case the optimal policy can be computed via dynamic programming using the Bellman recursion above.

\medskip
\noindent
In the setting considered in this work, however, the agent does not have access to the transition kernels and receives state information only at intermittent update times. Consequently, the optimal value function cannot be computed directly, and the decision-maker must instead learn and plan based on limited information.

\medskip
\noindent
We devote the remainder of this section to two steps. First, we develop a learning and planning framework that enables the agent to construct policies under the given information constraints. Second, we introduce a performance metric that quantifies the degradation incurred due to limited update rates, thereby providing a formal way to address the central question of this work.

\subsection{Algorithm Development}

We now describe a learning and planning framework that enables the agent to construct policies under the above information constraints. The proposed approach builds on the framework introduced in~\cite{ornik2021learning} and later extended to time-varying partially observable MDPs in~\cite{puthumanaillam2024weathering}. It combines likelihood-based estimation of time-varying transition kernels with finite-horizon planning.

\medskip
\noindent
The key difference from these prior works is that the agent has limited access to state information and performs estimation and planning only at the update times. Accordingly, the proposed approach consists of two main components: (i) estimation of the transition kernels using data collected at update times, and (ii) planning based on the estimated model to learn policies that are executed between consecutive updates.

\medskip
\noindent
In the remainder of this subsection, we describe these two components in detail and then present the resulting skip-update algorithm.

\subsubsection{Maximum Likelihood-based Estimation}

Let $\mathcal{D}_t$ denote the dataset of observed transitions up to time $t$,
\[
    \mathcal{D}_t
    =
    \{(s_{\tau_k}, a_{\tau_k}, s_{\tau_k+1}) : 0 \le k \le k(t)\},
\]
collected at update times. Given a sequence of transition kernels 
$\{P_{\tau_k}\}_{k=0}^{k(t)}$, the corresponding log-likelihood of the data is
\begin{equation*}
    \ell\big(\{P_{\tau_k}\}_{k=0}^{k(t)}; \mathcal{D}_t\big)
    =
    \sum_{k=0}^{k(t)}
    \log P_{\tau_k}(s_{\tau_k+1} \mid s_{\tau_k},a_{\tau_k}).
\end{equation*}
The maximum likelihood estimation of the transition kernels at time $t$ is defined as any maximizer of the constrained log-likelihood problem
\begin{equation}\label{eq:ccmle}
\begin{aligned}
    \max_{\{\hat{P}_{\tau_k}\}_{k=0}^{k(t)}} \quad
    & \sum_{k=0}^{k(t)}
    \log \hat{P}_{\tau_k}(s_{\tau_k+1} \mid s_{\tau_k},a_{\tau_k}) \\
    \text{s.t.}\quad
    & \sum_{s'\in\mathcal{S}} \hat{P}_{\tau_k}(s' \mid s, a) = 1,
    && \forall s \in \mathcal{S},\; a \in \mathcal{A},\; 0 \le k \le k(t), \\
    & \hat{P}_{\tau_k}(s' \mid s, a) \ge 0,
    && \forall s,s' \in \mathcal{S},\; a \in \mathcal{A},\; 0 \le k \le k(t), \\
    & \big|\hat{P}_{\tau_{k}}(s' \mid s, a)
          - \hat{P}_{\tau_{k-1}}(s' \mid s, a)\big|
    \le  \sum_{m=\tau_{k-1}}^{\tau_{k}-1} \varepsilon_{m},
    && \forall s,s' \in \mathcal{S},\; a \in \mathcal{A},\; 0 < k \leq k(t).
\end{aligned}
\end{equation}

\medskip
\noindent
The problem in~\eqref{eq:ccmle} is convex and Theorem~7 of~\cite{ornik2021learning} shows that the transition probabilities associated
with observed data, i.e., the triples $(s_{\tau_k},a_{\tau_k},s_{\tau_k+1}) \in \mathcal{D}_t$, are uniquely
determined, whereas the remaining components do not affect the objective and are only
constrained by affine feasibility conditions. Consequently, the solution set is a
bounded polyhedron (i.e., a polytope). Let this solution set be denoted by
$\mathcal{P}^{\mathcal{D}_{t}}$.

\medskip
\noindent
Let $\{\hat P_{\tau_k}\}_{k=0}^{k(t)}$ denote any particular solution returned by
\eqref{eq:ccmle}. Then $\mathcal{P}^{\mathcal{D}_{t}}$ is the set of all sequences of probability transition kernels that satisfy
\begin{equation}\label{eq:uncer}
\begin{aligned}
    &\tilde{P}_{\tau_k}(s_{\tau_k+1} \mid s_{\tau_k},a_{\tau_k})
    =
    \hat{P}_{\tau_k}(s_{\tau_k+1} \mid s_{\tau_k},a_{\tau_k}),
    && \forall\, 0 \le k \le k(t),\\
    &\tilde{P}_{\tau_k}(s' \mid s,a)
    \ge 0,
    && \forall s, s' \in \mathcal{S},\; a \in \mathcal{A},\; 0 \le k \le k(t),\\
    &\sum_{s' \in \mathcal{S}} \tilde{P}_{\tau_k}(s' \mid s,a)
    = 1,
    && \forall s \in \mathcal{S},\; a \in \mathcal{A},\; 0 \le k \le k(t),\\
    &\bigl|
    \tilde{P}_{\tau_{k}}(s' \mid s,a)
    -
    \tilde{P}_{\tau_{k-1}}(s' \mid s,a)
    \bigr|
    \le \sum_{m=\tau_{k-1}}^{\tau_{k}-1} \varepsilon_{m},
    && \forall s,s' \in \mathcal{S},\; a \in \mathcal{A},\; 0 < k \leq k(t).
\end{aligned}
\end{equation}
Now, for any state--action pair $(s,a)$, define
\begin{align*}
    \mathcal{U}^{\mathcal{D}_{t}}(s,a)
    :=
    \Bigl\{
    \tilde P_{\tau_{k(t)}}(\cdot\mid s,a)
    \;\Big|\;
    \{\tilde{P}_{\tau_k}\}_{k=0}^{k(t)} \in \mathcal{P}^{\mathcal{D}_{t}}
    \Bigr\}.
\end{align*}
We then quantify the uncertainty in the estimation of the probability transition kernels estimate by
\begin{align*}
    \mathrm{u}_{t}(s,a)
    :=
    \mathrm{diam}\bigl(\mathcal{U}^{\mathcal{D}_{t}}(s,a)\bigr).
\end{align*}
The construction above yields, at each time $t$, an estimated transition kernel together with an uncertainty set derived from $\mathcal{P}^{\mathcal{D}_t}$. 

\medskip
\noindent
Next, we describe how these objects are incorporated into a planning procedure to compute control policies.

\subsubsection{Planning with Estimated Transition Kernels}\label{sec:planning}

Planning is performed exclusively at the update times. Between update times, the agent plans and acts based on the most recently observed state, which serves as its internal state estimate. No model update or replanning is performed between consecutive update times.

\medskip
\noindent
In particular, at each update time $\tau_k$, given the estimated transition kernels $\{\hat{P}_{\tau_j}\}_{j=0}^{k}$ and the uncertainty measure $\mathrm{u}_{\tau_k}$ derived from the solution polytope based on the dataset $\mathcal{D}_{\tau_k}$, the agent solves a finite-horizon control problem. Let $H_k := \min\{\bar{H},\, T-1-\tau_k\}$ be the planning horizon at update time $\tau_{k}$ for some fixed $\bar{H} > 0$. The agent plans over a look-ahead horizon of $H_{k}$ by optimizing the expected cumulative reward predicted by the current model, while holding
the estimated transition kernel fixed over the planning horizon. The resulting policy is then executed until the next update time.

\medskip
\noindent
While the quantity $\mathrm{u}_{\tau_k}(s,a)$ characterizes the uncertainty in the transition dynamics at the current update time, finite-horizon planning requires uncertainty estimates at future times $\tau_k + h$, $h=0,\dots,H_k-1$, before additional data are collected. To this end, we construct a forecasted uncertainty measure $\mathrm{u}_{\tau_k+h \mid \tau_k}(s,a)$ using only
$\mathcal{D}_{\tau_k}$ by propagating the solution polytope of~\eqref{eq:ccmle} forward in time under the bounded drift model.

\medskip
\noindent
At update time $\tau_k$, the agent computes a control policy
$\pi^{\mathrm{alg}}_{\tau_k}$ by solving
\begin{align}\label{eq:policy-opt-finite}
    \pi^{\mathrm{alg}}_{\tau_k}
    \in
    \operatorname*{arg\,max}_{\{\pi_{h}\}_{h=0}^{H_k-1}}
    \mathbb{E}\!\left[
    \sum_{h=0}^{H_k-1}
    r^{(\beta)}_{\tau_k,h}\big(x_{h}, \pi_{h}(x_{h})\big)
    \right],
\end{align}
where the augmented reward is defined as
\[
    r^{(\beta)}_{\tau_k,h}(x,a)
    :=
    r_{\tau_k+h}(x,a)
    +
    \beta\, \mathrm{u}_{\tau_k+h\mid \tau_k}(x,a),
\]
with $\beta \ge 0$ controlling the weight placed on uncertainty. The state evolution over the planning horizon follows
\[
    x_{h+1}
    \sim
    \hat{P}_{\tau_k}\bigl(\cdot \mid x_h, \pi_h(x_h)\bigr),
\]
i.e., the estimated transition kernel at the most recent update time is held fixed throughout the planning horizon, consistent with the bounded-drift assumption and the absence of updates within the planning window.

\medskip
\noindent
Between two consecutive update times (i.e., during skip times), the agent does not update the model or replan. Instead, the policy computed at $\tau_k$ is applied in a receding-horizon manner for all times
\[
    t \in \{\tau_k,\tau_k+1,\ldots,\tau_{k+1}-1\}.
\]
This results in a skip-update model-predictive control strategy in which the implemented policy is piecewise constant between update times.

\medskip
\noindent
For the purpose of analysis, we can interpret the planning step
in~\eqref{eq:policy-opt-finite} as solving a finite-horizon optimal control problem for an auxiliary TVMDP with stationary transition kernel $\hat{P}_{\tau_k}$ and non-stationary rewards. By standard dynamic programming arguments, there exists an associated sequence of optimal value functions $\{\hat{W}^\star_{\tau_k,h}\}_{h=0}^{H_k}$ satisfying
\begin{equation}\label{eq:w_hat}
    \hat{W}^{\star}_{\tau_k,h}(s)
    :=
    \max_{\{\pi_{j}\}_{j=h}^{H_k-1}}
    \mathbb{E}\!\left[
    \sum_{j=h}^{H_k-1}
    r^{(\beta)}_{\tau_k,j}\!\big(x_j, \pi_j(x_j)\big)
    \;\Big|\; x_h = s
    \right],
\end{equation}
with terminal condition $\hat{W}^\star_{\tau_k,H_k} \equiv 0$.

\medskip
\noindent
By combining these two components, we obtain the skip-update Algorithm~\ref{al:main} operating under limited update rate.

\begin{algorithm}[t]
\caption{Skip-update Algorithm with Limited Update Rate}\label{al:main}
\begin{algorithmic}[1]
\STATE \textbf{Input:} horizon $T$, drift bounds $\{\varepsilon_t\}_{t=0}^{T-1}$, planning cap $\bar H$, exploration weight $\beta\ge 0$
\STATE Initialize dataset $\mathcal{D}_{-1} \gets \emptyset$
\STATE Set $k \gets 0$, $\tau_{0} \gets 0$
\STATE Observe initial state $s_{0}$ 
\STATE Set $\hat{s}\gets s_{0}$ \hfill\COMMENT{most recent observed state}
\STATE Initialize policy $\pi^{\mathrm{alg}}_{0}$ arbitrarily

\FOR{$t = 0$ to $T-1$}

    \IF{$t \in \mathcal{T}_{\mathrm{upd}}$}
    
        \STATE Observe $s_t$
        \STATE Set $\tau_k \gets t$
        
        \STATE Select action $a_t \sim \pi^{\mathrm{alg}}_{t}(s_t)$, apply $a_t$, observe $s_{t+1}$, and receive reward $r_t$
        \STATE Append $(s_t, a_t, s_{t+1})$ to dataset: $\mathcal{D}_{t} \gets \mathcal{D}_{t-1} \cup \{(s_t, a_t, s_{t+1})\}$
        
        \STATE Solve~\eqref{eq:ccmle} using $\mathcal{D}_t$ to obtain an estimator
        $\{\hat P_{\tau_j}\}_{j=0}^{k}$ and solution polytope $\mathcal{P}^{\mathcal{D}_t}$
        
        \STATE Compute uncertainty sets $\mathcal{U}^{\mathcal{D}_{t}}(s,a)$ and diameters
        $\mathrm{u}_{t}(s,a)=\mathrm{diam}\!\left(\mathcal{U}^{\mathcal{D}_{t}}(s,a)\right)$
        
        \STATE Set planning horizon $H_t \gets \min\{\bar H,\, T-1-t\}$
        
        \STATE Construct forecasted uncertainty $\mathrm{u}_{t+h\mid t}(s,a)$ for $h=0,\ldots,H_t-1$
        
        \STATE Solve~\eqref{eq:policy-opt-finite} to obtain an $H_t$-step plan
        $\{\pi^{\mathrm{plan}}_{t,h}\}_{h=0}^{H_t-1}$ with
        \[
            r^{(\beta)}_{t,h}(s,a)=r_{t+h}(s,a)+\beta\,\mathrm{u}_{t+h\mid t}(s,a),\quad
            x_{h+1}\sim \hat P_{\tau_k}(\cdot\mid x_h,\pi^{\mathrm{plan}}_{t,h}(x_h))
        \]
        
        \STATE Set $\pi^{\mathrm{alg}}_{t} \gets \pi^{\mathrm{plan}}_{t,0}$ \hfill\COMMENT{MPC: apply first policy}
        \STATE Set $\hat{s}\gets s_{t+1}$
        \STATE Set $k \gets k + 1$

    \ELSE
    
        \STATE Set $\pi^{\mathrm{alg}}_{t} \gets \pi^{\mathrm{alg}}_{\tau_k}$ \hfill\COMMENT{reuse most recent policy}
        \STATE Select action $a_t \sim \pi^{\mathrm{alg}}_{t}(\hat{s})$, apply $a_t$, and receive reward $r_t$
        \STATE Set $\mathcal{D}_{t} \gets \mathcal{D}_{t-1}$ \hfill\COMMENT{no new data}
        \STATE (No new state observation; keep $\hat{s}$ unchanged)
        
    \ENDIF

\ENDFOR
\end{algorithmic}
\end{algorithm}

\subsection{Dynamic Regret Objective}\label{sec:regret}

Having specified the learning and planning framework, we now evaluate its performance through a dynamic regret criterion.

\medskip
\noindent
We adopt a dynamic regret formulation that compares the cumulative reward achieved by the proposed algorithm to that of an optimal policy with full knowledge of the time-varying transition kernels and continuous access to the system state. This metric quantifies how performance degrades when the decision-maker operates with intermittently updated information and under time-varying dynamics.

\medskip
\noindent
Recall that the agent does not have access to state information at every time step.
Let
\[
    \tilde{s}_t =
        \begin{cases}
        s_t, & t \in \mathcal{T}_{\mathrm{upd}}, \\
        s_{\tau_{k(t)}}, & t \in \mathcal{T}_{\mathrm{skip}}.
        \end{cases}
\]
denote the state information available to the agent at time $t$, i.e., the true state at update times and the most recently observed state otherwise.

\medskip
\noindent
The cumulative reward achieved by Algorithm~\ref{al:main} is defined as
\begin{equation}\label{eq:alg-cost}
    J_T^{\mathrm{alg}}(s)
    =
    \mathbb{E}\!\left[
        \sum_{t=0}^{T-1}
        r_t\big(s_t, \pi_t^{\mathrm{alg}}(\tilde s_t)\big)
        \,\middle|\,
        s_{0} = s
    \right],\quad \forall s \in \mathcal{S},
\end{equation}
where the state trajectory evolves under the true time-varying transition kernels $\{P_t\}_{t=0}^{T-1}$.

\medskip
\noindent
The dynamic regret of Algorithm~\ref{al:main} is defined as
\begin{equation}\label{eq:dyn-reg}
    \mathcal{DR}(T)
    =
    \max_{s \in \mathcal{S}}
    \Big(
        J_T^\star(s) - J_T^{\mathrm{alg}}(s)
    \Big).
\end{equation}

\medskip
\noindent
This regret reflects the combined effects of temporal variation in the dynamics, estimation uncertainty, and the limited update structure that forces the agent to act based on partial and stale information.

\medskip
\noindent
Our objective is to establish an upper bound on $\mathcal{DR}(T)$ of the skip-update Algorithm~\ref{al:main}. The detailed analysis is presented in the next section.

\section{Regret Analysis}
\label{sec:regret}

We now turn to answering our central question: how the performance of the decision-maker degrades when it is forced to operate for extended periods without new information or updates. To this end, we establish an upper bound on the dynamic regret defined in Section~\ref{sec:regret}. The resulting bound characterizes the interplay between three key factors in our setting: the temporal variation of the transition kernels, the uncertainty arising from limited observations, and the duration between successive updates during which the agent must act based on stale information.

\medskip
\noindent
Due to the intermittent information structure, we decompose the dynamic regret into two components: (i) errors incurred at update times, where the agent performs estimation and planning, and (ii) errors accumulated during skip intervals, where the agent operates using stale information. The former follows similar arguments for receding-horizon planning with model mismatch~\cite{zhang2024predictive}, while the latter constitutes the main focus of our analysis.

Due to the intermittent information structure, we decompose the dynamic regret into two components: (i) errors incurred at update times, where the agent performs estimation and planning, and (ii) errors accumulated during skip intervals, where the agent operates using stale information. The former follows standard arguments for receding-horizon planning with model mismatch~\cite{zhang2024predictive}, while the latter constitutes the main focus of our analysis.

\medskip
\noindent
Recall the dynamic regret defined in~\eqref{eq:dyn-reg}. By a standard telescoping argument (see Appendix~\ref{app_0}), it can be rewritten as
\begin{equation}\label{eq:regret_decomposition}
    \mathcal{DR}(T)
    =
    \max_{s_0 \in \mathcal{S}}
    \left(
    \sum_{t \in \mathcal{T}_{\mathrm{upd}}} \Delta_t(s_0)
    +
    \sum_{t \in \mathcal{T}_{\mathrm{skip}}} \Delta_t(s_0)
    \right),
\end{equation}
where $\Delta_t(s_0)$ denotes the contribution of time $t$ to the dynamic regret along the trajectory starting from $s_0$. The first term corresponds to update times, while the second term captures the accumulation of error during skip intervals. Characterizing how this second term grows with the length of skip intervals and the rate of temporal variation constitutes the main contribution of this work.

\medskip
\noindent
To control how the errors in~\eqref{eq:regret_decomposition} propagate over time, we impose a finite-time contractiveness condition on the time-varying transition kernels. This condition will be used to bound both the update and skip contributions, and is particularly important for controlling error accumulation during skip intervals. We formalize this condition through
the following overlap coefficient.

\begin{definition}[Overlap coefficient]
Fix $m \ge 1$ and consider two sequence of policies $\{\pi^1_{t}\}_{t \ge 0}$ and $\{\pi^2_{t}\}_{t \ge 0}$.
For any $t \ge 0$, define the overlap coefficient of these policies as
\begin{equation}\label{eq:eta-timevarying}
\begin{aligned}
    \eta_t(\pi^1,\pi^2)
    :=
    \min_{s^{1},s^{2} \in \mathcal{S}}
    \sum_{s' \in \mathcal{S}}
    \min \Bigl\{
        P_{t:t+m-1}^{\pi^1_{t:t+m-1}}(s' \mid s^{1}),
        P_{t:t+m-1}^{\pi^2_{t:t+m-1}}(s' \mid s^{2})
    \Bigr\},
\end{aligned}
\end{equation}
where $P_{t:t+m-1}^{\pi_{t:t+m-1}}$ denotes the $m$-step transition kernel over $[t,t+m)$ induced by $\{\pi_{t}\}_{t \ge 0}$.
\end{definition}

\medskip
\noindent
By construction, $\eta_t(\pi^1,\pi^2) \ge 0$. To ensure sufficient contraction
relative to the optimal policy, we impose the following assumption.

\begin{assumption}[Finite-time mixing relative to the optimal policy]\label{ass:mixing}
There exist constants $m \ge 1$ and $\eta > 0$ such that for the optimal policy
$\{\pi_t^\star\}_{t=0}^{T-1}$ and the policy generated by
Algorithm~\ref{al:main}, denoted $\pi^{\mathrm{alg}}$, the overlap coefficient
satisfies
\[
    \eta_t(\pi^{\mathrm{alg}},\pi^\star) \ge \eta,
    \qquad \forall t \in \{0,\ldots,T-1\}.
\]
\end{assumption}

\medskip
\noindent
Assumption~\ref{ass:mixing} ensures that the multi-step state distributions
induced by the algorithm and the optimal policy retain a uniform amount of
overlap over any window of length $m$. This finite-time contractiveness
property is standard in nonstationary MDP analysis and is crucial for
controlling the propagation of model estimation and planning errors in the
dynamic regret analysis.

\medskip
\noindent
We next introduce several quantities that appear in the regret bound. Let
\[
    \tilde{V} := \max_{k} \operatorname{sp} (V_{k}^{\star}),
    \qquad
    \tilde{W}_{t} := \max_{k} \operatorname{sp} (\hat{W}_{t,k}^{\star}),
\]
denote uniform bounds on the span seminorm of the optimal value functions and the auxiliary value functions arising in the planning step. To quantify estimation and planning errors at update times, define
\[
    \hat{\varepsilon}_{t,i}
    :=
    \max_{(s,a)} \bigl\|P_{t+i}(\cdot \mid s,a)-\hat{P}_{t}(\cdot \mid s,a)\bigr\|_{\operatorname{tv}},
    \qquad
    \hat{\delta}_{t,i}
    :=
    \beta \max_{(s,a)} \mathrm{u}_{t+i \mid t}(s,a).
\]
In addition, to quantify the effect of temporal variation during skip intervals, define
\[
    \bar{\varepsilon}_{\tau_{k(t)},t}
    :=
    \max_{s,a} \bigl\| P_{t}(\cdot\mid s,a) - P_{\tau_{k(t)}}(\cdot\mid s,a)\bigr\|_{\mathrm{tv}},\quad \bar{\delta}_{\tau_{k(t)},t}
    :=
    \max_{s} \mathrm{sp} \big( r_{t}(s,\cdot) - r_{\tau_{k(t)}}(s,\cdot) \big).
\]
We now present the dynamic regret bound for Algorithm~\ref{al:main}.

\begin{theorem}\label{th:main}
Suppose Assumption~\ref{ass:mixing} holds. Then Algorithm~\ref{al:main} achieves the dynamic regret
\begin{equation}\label{eq:main_regret}
    \mathcal{DR}(T)
    \le
    \sum_{t=0}^{T-1}
    \left(
    \alpha^{\left\lfloor \frac{H_{\tau_{k(t)}}-1}{m} \right\rfloor} \tilde{V}
    +
    \hat{E}_{\tau_{k(t)}}
    \right)
    +
    \sum_{t \in \mathcal{T}_{\mathrm{skip}}}
    \left(
    \alpha^{\left\lfloor \frac{T-t}{m} \right\rfloor} \tilde{V}
    +
    \bar{E}_t
    \right),
\end{equation}
where $\hat{E}_t$ and $\bar{E}_t$ are defined as:
\begin{align*}
    \hat{E}_{t}
    &= \hat{\delta}_{t,0} + \hat{\varepsilon}_{t,0}\ \tilde{V} 
    + 2 \sum_{\ell = 0}^{\lfloor \frac{H_{t}-1}{m} \rfloor - 1} \alpha^{\ell} \sum_{i=1}^{m} \hat{e}_{t,\ell m+i}  + 2 \alpha^{\lfloor \frac{H_{t}-1}{m} \rfloor} \sum_{i=1}^{H_{t}-\lfloor \frac{H_{t}-1}{m} \rfloor m - 1} \hat{e}_{t,\lfloor \frac{H_{t}-1}{m} \rfloor m+i}, \\
    \bar{E}_{t}
    &= \bar{e}_{\tau_{k(t)}, t-\tau_{k(t)}} 
    + 2 \sum_{\ell =0}^{\lfloor \frac{T-t}{m}\rfloor - 1} \alpha^{\ell }\ \sum_{i=0}^{m-1} \bar{e}_{\tau_{k(t)}, \ell m+i}+ 2 \alpha^{\lfloor \frac{T-t}{m}\rfloor}\ \sum_{i=0}^{T-t - \lfloor\frac{T-t}{m}\rfloor m } \bar{e}_{\tau_{k(t)}, \lfloor \frac{T-t}{m}\rfloor m+i},
\end{align*}
with $\hat{e}_{t, j} = \hat{\varepsilon}_{t,j}\ \tilde{W}_{t} + \hat{\delta}_{t,j}$ and $\bar{e}_{t,j} = \bar{\varepsilon}_{\tau_{k(t)}, j}\ \tilde{V} + \bar{\delta}_{\tau_{k(t)}, j}$.
\end{theorem}

\medskip
\noindent
The proof is provided in detail in Appendix~\ref{app_0}. We now interpret the terms appearing in Theorem~\ref{th:main}. The regret bound in~\eqref{eq:main_regret} decomposes into two contributions: errors incurred at update times and errors accumulated during skip intervals.
\begin{itemize}
    \item \textbf{Update-time error.}
    The term
    $\alpha^{\left\lfloor \frac{H_{\tau_{k(t)}}-1}{m}\right\rfloor}\tilde V$ reflects the use of finite-horizon planning within a model predictive control framework, where optimizing over a truncated horizon and executing only the first action introduces a mismatch with the full-horizon optimal policy. Under Assumption~\ref{ass:mixing}, this mismatch contracts over blocks of length $m$. The term $\hat{E}_{\tau_{k(t)}}$ captures errors arising from planning with an estimated model, including both statistical estimation error and mismatch with the true time-varying dynamics. These errors are also attenuated by the contraction property.
    \item \textbf{Skip-interval error.}
    The regret incurred during skip intervals arises from a recursive error propagation combined with mismatch effects induced by stale information. The instantaneous regret at a skip time $t \in \mathcal{T}_{\mathrm{skip}}$ is decomposed as
    \[
    \Delta_t
    =
    \Delta_{\tau_{k(t)}}
    +
    \text{(time mismatch)}
    +
    \text{(state mismatch)}.
    \]    
    The first term, $\Delta_{\tau_{k(t)}}$, shows that regret at time $t$ inherits the error from the most recent update time, reflecting that no correction is performed during skip intervals. The second term captures time mismatch, i.e., the discrepancy between the optimal decision problems at time $t$ and at the last update time $\tau_{k(t)}$, arising from temporal drift in both transition kernels and rewards:
    \[
    \bar{\varepsilon}_{\tau_{k(t)},t}
    =
    \max_{s,a}
    \|P_t(\cdot|s,a) - P_{\tau_{k(t)}}(\cdot|s,a)\|_{\mathrm{tv}},
    \quad
    \bar{\delta}_{\tau_{k(t)},t}
    =
    \max_s \mathrm{sp}\big(r_t(s,\cdot) - r_{\tau_{k(t)}}(s,\cdot)\big).
    \]
    The third term captures state mismatch, which arises because the policy is evaluated at the stale state $s_{\tau_{k(t)}}$ instead of the current state $s_t$.
    
    These components jointly generate the skip-interval error described in the theorem. The quantities $\bar{\varepsilon}_{\tau_{k(t)},\cdot}$ and $\bar{\delta}_{\tau_{k(t)},\cdot}$ accumulate over time and give rise to the term $\bar{E}_t$, which captures how discrepancies grow with the length of the skip interval. At the same time, their impact on future performance is not uniform: by the multi-step contraction property (Lemma~\ref{lem:multi_stage_error}), the influence of errors incurred at time $t$ decays geometrically over time, yielding the factor $\alpha^{\lfloor (T-t)/m \rfloor}\tilde V$.
    
    Overall, the regret during skip intervals is governed by the interplay between error accumulation, driven by temporal drift and stale decision-making, and error propagation, controlled by the contraction coefficient $\alpha$. This structure explains how longer skip intervals and faster temporal variation increase regret, while stronger mixing mitigates its long-term effect.  
\end{itemize}

\section*{Acknowledgment}
This work was supported by the Office of Naval Research under grant no. N00014-25-1-2369.

\bibliographystyle{plain}
\bibliography{references}

\appendices
\section{Helper Lemma}\label{app_00}

We state a helper lemma that quantifies how the difference between two optimal value functions propagates over multiple stages under the mixing assumption. This result is adapted from \cite{zhang2024predictive}, and will be used to control the error between value functions associated with different models. 

\medskip
\noindent
Let the one-stage optimal Bellman operators $\mathcal{B}_t$ and $\bar{\mathcal{B}}_t$ be defined as
\begin{align*}
    (\mathcal{B}_t f)(s)
    &:=
    \max_{a \in \mathcal{A}}
    \Big\{
        r_t(s,a)
        +
        \mathbb{E}_{s' \sim P_t(\cdot \mid s,a)}[f(s')]
    \Big\},\\
    (\bar{\mathcal{B}}_t f)(s)
    &:=
    \max_{a \in \mathcal{A}}
    \Big\{
        \bar{r}_t(s,a)
        +
        \mathbb{E}_{s' \sim \bar{P}_t(\cdot \mid s,a)}[f(s')]
    \Big\},
\end{align*}
for all $s \in \mathcal{S}$ and any $f:\mathcal{S} \rightarrow \mathbb{R}$. The corresponding optimal value functions $\{V_t^\star\}_{t=0}^{T}$ and $\{\bar{V}_t^\star\}_{t=0}^{T}$ satisfy
\[
    V_t^\star = \mathcal{B}_t V_{t+1}^\star,
    \qquad
    \bar{V}_t^\star = \bar{\mathcal{B}}_t \bar{V}_{t+1}^\star,
    \qquad t = 0,\dots,T-1,
\]
with $V_T^\star \equiv 0$ and $\bar{V}_T^\star \equiv 0$.

\begin{lemma}[Multi-stage error propagation under mixing {\cite{zhang2024predictive}}]
\label{lem:multi_stage_error}
Suppose Assumption~\ref{ass:mixing} holds with constants $m$ and $\eta$, and define $\alpha = 1 - \eta$. Then, for any integer $N \geq 1$ such that $N \leq T - t$, the following holds:
\begin{align*}
    \mathrm{sp} \big( V^{\star}_{t} - \bar{V}^{\star}_{t} \big)
    &\leq  
    \alpha^{\lfloor \frac{N}{m} \rfloor}  \ \mathrm{sp} \Big(  V^{\star}_{t+N} - \bar{V}^{\star}_{t+N}\Big)\\
    &\quad + 2 \alpha^{\lfloor \frac{N}{m} \rfloor}
    \sum_{i=1}^{N-\lfloor \frac{N}{m} \rfloor m - 1}
    \Big(
        \bar{\varepsilon}_{t+\lfloor \frac{N}{m} \rfloor m+i}\ 
        \mathrm{sp} \big(  \bar{V}^{\star}_{t+(\lfloor \frac{N}{m} \rfloor )m+1} \big)
        + \bar{\delta}_{t+\lfloor \frac{N}{m} \rfloor m+i}
    \Big)\\
    &\quad + 2 \sum_{\ell = 0}^{\lfloor \frac{N}{m} \rfloor - 1}
    \alpha^{\ell}
    \sum_{i=1}^{m}
    \Big(
        \bar{\varepsilon}_{t+\ell m+i}\ 
        \mathrm{sp} \big(  \bar{V}^{\star}_{t+(\ell+1)m+1} \big)
        + \bar{\delta}_{t+\ell m+i}
    \Big).
\end{align*}
Here,
\[
\bar{\varepsilon}_{k} = \max_{(s,a)} \|P_{k}(\cdot \mid s,a)-\bar{P}_{k}(\cdot \mid s,a)\|_{\operatorname{tv}},
\qquad
\bar{\delta}_{k} = \max_{(s,a)} |r_{k}(s,a)-\bar{r}_{k}(s,a)|.
\]
\end{lemma}

\noindent
This lemma shows that the span difference between value functions contracts geometrically over blocks of length $m$, up to additive terms that capture discrepancies in transition kernels and rewards. The proof follows directly from \cite{zhang2024predictive} and is omitted here.

\section{Proof of the Theorem~\ref{al:main}}\label{app_0}

In this appendix, we prove Theorem~\ref{th:main}. We relate the dynamic regret to
(i) finite-time contraction under Assumption~\ref{ass:mixing},
(ii) estimation and bonus errors arising from planning with estimated probability transitions,
(iii) finite planning horizon effects,
(iv) pauses in updates due to limited information and computation, and
(v) reuse of stale policies during skip intervals. For this, let us start with some definitions and notation.

\paragraph{Value functions.}
Recall the optimal value function of the TVMDP $\mathcal{M}$ at time $t$,
\[
    V_t^{\star}(s)
    :=
    \max_{\{\pi_k\}_{k=t}^{T-1}}
    \mathbb{E}\!\left[
        \sum_{k=t}^{T-1} r_k\big(s_k,\pi_k(s_k)\big)
        \,\middle|\, s_t = s
    \right],\qquad \forall s \in \mathcal{S}
\]
with $V^{\star}_{T} \equiv 0$, which satisfies
\[
    V_t^{\star}(s)
    =
    \max_{a \in \mathcal{A}}
    \left\{
        r_t(s,a)
        +
        \mathbb{E}_{s' \sim P_t(\cdot \mid s,a)}
        \big[ V_{t+1}^{\star}(s') \big]
    \right\},\qquad \forall s \in \mathcal{S}.
\]
Define
\[
    Q_t^\star(s,a):= r_t(s,a)+\mathbb{E}_{s' \sim P_t(\cdot \mid s,a)}[V_{t+1}^{\star}(s')],\qquad \forall (s, a) \in \mathcal{S} \times \mathcal{A}
\]
then $V_t^\star(s)=\max_{a\in\mathcal{A}}Q_t^\star(s,a)$ for any $s \in \mathcal{S}$.

\medskip
\noindent
Similarly, the value function of the policy implemented by Algorithm~\ref{al:main} is
\[
    V_t^{\mathrm{alg}}(s)
    :=
    \mathbb{E}\!\left[
        \sum_{k=t}^{T-1} r_k\big(s_k,\pi_k^{\mathrm{alg}}(s_k)\big)
        \,\middle|\, s_t = s
    \right], \qquad \forall s \in \mathcal{S}.
\]
Hence, the dynamic regret of Algorithm~\ref{al:main} defined in~\eqref{eq:dyn-reg} can be expressed as
\[
    \mathcal{DR}(T)=\max_{s_0\in\mathcal{S}}\big(V_0^\star(s_0)-V_0^{\mathrm{alg}}(s_0)\big).
\]

\paragraph{Telescoping decomposition.}
For ease of notation, let us denote $a_{t}^\star = \arg\max_{a}Q_{t}^\star(s_{t},a)$ and $a_{t}^{\mathrm{alg}} \sim \pi_{t}^{\mathrm{alg}}(s_{\tau_{k(t)}})$. Now we can rewrite $V_0^\star(s_0)-V_0^{\mathrm{alg}}(s_0)$ as
\begin{align*}
    V_0^\star(s_0)-V_0^{\mathrm{alg}}(s_0)
    &= 
    r_0(s_{0},a_{0}^{\star})
    +
    \mathbb{E}_{s_{1} \sim P_0(\cdot \mid s_0,a_0^\star)}
    \big[ V_{1}^{\star}(s_{1}) \big] - r_0(s_{0},a_{0}^{\mathrm{alg}})
    -
    \mathbb{E}_{s_{1} \sim P_0(\cdot \mid s_0,a_{0}^{\mathrm{alg}})}
    \big[ V_{1}^{\mathrm{alg}}(s_{1}) \big]\\
    &+ \mathbb{E}_{s_{1} \sim P_0(\cdot \mid s_0,a_{0}^{\mathrm{alg}})}
    \big[ V_{1}^{\star}(s_{1}) \big] - \mathbb{E}_{s_{1} \sim P_0(\cdot \mid s_0,a_{0}^{\mathrm{alg}})}
    \big[ V_{1}^{\star}(s_{1}) \big]\\
    &= 
    Q_0^\star(s_0,a_0^\star)-Q_0^\star(s_0,a_0^{\mathrm{alg}})
    + \mathbb{E}_{s_1\sim P_0(\cdot\mid s_0,a_0^{\mathrm{alg}})}
    \big[V_1^\star(s_1)-V_1^{\mathrm{alg}}(s_1)\big],
\end{align*}        
Iterating for $t=0,\dots,T-1$ yields
\begin{equation}\label{eq:V0_gap}
\begin{aligned}
    V_0^{\star}(s_0) - V_0^{\mathrm{alg}}(s_0)
    &=
    Q_{0}^{\star}(s_{0},a_{0}^\star) - Q_{0}^{\star}(s_{0},a_{0}^{\mathrm{alg}})\\
    &\quad + \sum_{k=1}^{T-1} \mathbb{E}_{s_{1}\sim P_{0}(\cdot \mid s_{0},a_{0}^{\mathrm{alg}})}\ \dots \mathbb{E}_{s_{k}\sim P_{k-1}(\cdot \mid s_{k-1},a_{k-1}^{\mathrm{alg}})} 
    \!\big[ Q_{k}^{\star}(s_{k},a_{k}^\star) - Q_{k}^{\star}(s_{k},a_{k}^{\mathrm{alg}}) \big],  
\end{aligned}
\end{equation}

\medskip
\noindent
Now define $\Delta_t(s_{0})$ as
\begin{equation}\label{eq:per-step-regret}
    \Delta_t(s_{0}) \;:=\; \mathbb{E}_{s_{1}\sim P_{0}(\cdot \mid s_{0},a_{0}^{\mathrm{alg}})}\ \dots \mathbb{E}_{s_{t}\sim P_{t-1}(\cdot \mid s_{t-1},a_{t-1}^{\mathrm{alg}})} 
    \!\big[ Q_{t}^{\star}(s_{t},a_{t}^\star) - Q_{t}^{\star}(s_{t},a_{t}^{\mathrm{alg}}) \big],\qquad \forall t\geq 1,
\end{equation}
and $\Delta_{0}(s_{0}) := Q_{0}^{\star}(s_{0},a_{0}^\star) - Q_{0}^{\star}(s_{0},a_{0}^{\mathrm{alg}})$.  then we can rewrite the dynamic regret as
\begin{equation}
    \mathcal{DR}(T)
    \;=\;
    \max_{s_{0} \in \mathcal{S}}\ \sum_{t=0}^{T-1} G_t(s_{0})
    \;=\;
    \max_{s_{0} \in \mathcal{S}}\Bigg\{ \underbrace{\sum_{t\in\mathcal{T}_{\mathrm{upd}}} \Delta_t(s_{0})}_{\text{update steps}}
    \;+\;
    \underbrace{ \sum_{t\in\mathcal{T}_{\mathrm{skip}}} \Delta_t(s_{0})}_{\text{skip steps}} \Bigg\}.
\label{eq:two-way-decomp}
\end{equation}

\medskip
\noindent
Next, we state two lemmas: one for bounding the optimal $Q$-value difference at update times, i.e. $\Delta_{t}(s_{0})$ for $t \in \mathcal{T}_{\mathrm{upd}}$, and another for bounding the optimal $Q$-value difference at skip times, i.e. $\Delta_{t}(s_{0})$ for $t \in \mathcal{T}_{\mathrm{skip}}$.

\begin{lemma}\label{lem:Q_value_diff_obs}
Suppose $t \in \mathcal{T}_{\mathrm{upd}}$, then $G_{t}(s_{0})$ can be bounded as:
\begin{align*}
    \Delta_{t}(s_{0})
    &\leq 
    \alpha^{\lfloor \frac{H_{t}-1}{m} \rfloor}  \ \tilde{V} + \hat{\delta}_{t,0} + \hat{\varepsilon}_{t,0}\ \tilde{V} 
    + 2 \sum_{\ell = 0}^{\lfloor \frac{H_{t}-1}{m} \rfloor - 1} \alpha^{\ell} 
    \sum_{i=1}^{m} \Big( \hat{\varepsilon}_{t,\ell m+i}\ \tilde{W}_{t} + \hat{\delta}_{t,\ell m+i}  \Big)\\
    &\quad + 2 \alpha^{\lfloor \frac{H_{t}-1}{m} \rfloor} 
    \sum_{i=1}^{H_{t}-\lfloor \frac{H_{t}-1}{m} \rfloor m - 1} 
    \Big( \hat{\varepsilon}_{t,\lfloor \frac{H_{t}-1}{m} \rfloor m+i}\ \tilde{W}_{t} 
    + \hat{\delta}_{t,\lfloor \frac{H_{t}-1}{m} \rfloor m+i}  \Big),
\end{align*}
where $\tilde{V} = \max_{k} \operatorname{sp} (V_{k}^{\star})$, 
$\tilde{W}_{t} = \max_{k} \operatorname{sp} (\hat{W}_{t,k}^{\star})$,  
$\hat{\varepsilon}_{t,i} := \max_{(s,a)} \|P_{t+i}(\cdot \mid s,a)-\hat{P}_{t}(\cdot \mid s,a)\|_{\operatorname{tv}}$, 
and $\hat{\delta}_{t,i} := \beta  \max_{(s,a)} \mathrm{u}_{t+i \mid t}(s,a)$. 
\end{lemma}

\medskip
\noindent
The proof of this lemma is provided at Appendix~\ref{app_1}. We now present the following result, which constitutes the main contribution.

\begin{lemma}\label{lem:Q_value_diff_skip}
Suppose $t \in \mathcal{T}_{\mathrm{skip}}$ and $\tau_{k(t)}$ is the most recent update time. Then $\Delta_{t}(s_{0})$ can be bounded as:
\begin{align*}
    \Delta_{t}(s_{0}) 
    &\leq G_{\tau_{k(t)}}(s_0) 
    + \alpha^{\lfloor \frac{T-t}{m}\rfloor}\ \tilde{V} 
    + \bar{\delta}_{\tau_{k(t)}, t} 
    + \bar{\varepsilon}_{\tau_{k(t)}, t} \tilde{V}\\
    &\quad +2 \alpha^{\lfloor \frac{T-t}{m}\rfloor}\ 
    \sum_{i=0}^{T-t - \lfloor\frac{T-t}{m}\rfloor m } 
    \bigg( 
    \bar{\varepsilon}_{\tau_{k(t)}, \tau_{k(t)}+\lfloor \frac{T-t}{m}\rfloor m+i}\ \tilde{V} 
    + \bar{\delta}_{\tau_{k(t)}, \tau_{k(t)}+\lfloor \frac{T-t}{m}\rfloor m+i}  
    \bigg)\\
    &\quad + 2 \sum_{\ell =0}^{\lfloor \frac{T-t}{m}\rfloor - 1} \alpha^{\ell }\ 
    \sum_{i=0}^{m-1} 
    \bigg( 
    \bar{\varepsilon}_{\tau_{k(t)}, \tau_{k(t)}+\ell m+i}\ \tilde{V} 
    + \bar{\delta}_{\tau_{k(t)}, \tau_{k(t)}+\ell m+i}  
    \bigg).
\end{align*}
where 
$\bar{\varepsilon}_{\tau_{k(t)},t} := \max_{s,a} \| P_{t}(\cdot\mid s,a) - P_{\tau_{k(t)}}(\cdot\mid s,a)\|_{\mathrm{tv}}$, 
and 
$\bar{\delta}_{\tau_{k(t)},t} := \max_{s} \mathrm{sp} \big( r_{t}(s,\cdot) - r_{\tau_{k(t)}}(s,\cdot) \big)$.
\end{lemma}

\medskip
\noindent
The proof of this lemma is provided at Appendix~\ref{app_2}.

\noindent
By combining the results of these two lemmas and summing over the whole time horizon, we obtain:
\begin{align*}
    \mathcal{DR}(T) 
    &\leq  \sum_{t=0}^{T-1} 
    \Big( 
    \alpha^{\big \lfloor \frac{H_{\tau_{k(t)}}-1}{m} \big \rfloor}  \ \tilde{V} 
    + \hat{E}_{\tau_{k(t)}} 
    \Big)
    + \sum_{t \in \mathcal{T}_{\mathrm{skip}}} 
    \Big( 
    \alpha^{\lfloor \frac{T-t}{m}\rfloor}\ \tilde{V} 
    + \bar{E}_{t} 
    \Big)
\end{align*}
where $\hat{E}_{t}$ and $\bar{E}_{t}$ are defined as:
\begin{align*}
    \hat{E}_{t}
    &= \hat{\delta}_{t,0} + \hat{\varepsilon}_{t,0}\ \tilde{V} 
    + 2 \sum_{\ell = 0}^{\lfloor \frac{H_{t}-1}{m} \rfloor - 1} \alpha^{\ell} 
    \sum_{i=1}^{m} \hat{e}_{t,\ell m+i}  + 2 \alpha^{\lfloor \frac{H_{t}-1}{m} \rfloor} 
    \sum_{i=1}^{H_{t}-\lfloor \frac{H_{t}-1}{m} \rfloor m - 1} 
    \hat{e}_{t,\lfloor \frac{H_{t}-1}{m} \rfloor m+i}, \\
    \bar{E}_{t}
    &= \bar{e}_{\tau_{k(t)}, t-\tau_{k(t)}} 
    + 2 \sum_{\ell =0}^{\lfloor \frac{T-t}{m}\rfloor - 1} \alpha^{\ell }\ 
    \sum_{i=0}^{m-1} \bar{e}_{\tau_{k(t)}, \ell m+i}
    + 2 \alpha^{\lfloor \frac{T-t}{m}\rfloor}\ 
    \sum_{i=0}^{T-t - \lfloor\frac{T-t}{m}\rfloor m } 
    \bar{e}_{\tau_{k(t)}, \lfloor \frac{T-t}{m}\rfloor m+i},
\end{align*}
with 
$\hat{e}_{t, j} = \hat{\varepsilon}_{t,j}\ \tilde{W}_{t} + \hat{\delta}_{t,j}$ 
and  
$\bar{e}_{t,j} = \bar{\varepsilon}_{\tau_{k(t)}, j}\ \tilde{V} + \bar{\delta}_{\tau_{k(t)}, j}$. 
This concludes the proof.
\section{Analysis of Regret at Update Steps}\label{app_1}

\begin{proof}[Proof of Lemma~\ref{lem:Q_value_diff_obs}]
We start the proof by introducing an auxiliary finite-horizon TVMDP at any update time $t \in \mathcal{T}_{\mathrm{upd}}$, $\hat{\mathcal{M}}_{t}$ as:
\begin{align*}
    \hat{\mathcal{M}}_{t} = \Big(\mathcal{S}, \mathcal{A}, \hat{P}_{t}, \{r^{(\beta)}_{t,h}\}_{h=0}^{H_{t}-1}, H_{t}\Big)
\end{align*}
with augmented reward  
\[
    r^{(\beta)}_{t,h}(s,a) = r_{t+h}(s,a) + \beta \mathrm{u}_{t+h \mid t}(s,a), \qquad h=0,\dots,H_{t}-1,
\]  
for any $(s,a) \in \mathcal{S} \times \mathcal{A}$ and transition kernel $\hat P_{t}$, which is fixed for the whole horizon $H_{t} = \min\{\bar{H}, T-t\}$.

\medskip
\noindent
We denote its value function by $\hat W^{\star}_{t,h}$ as introduced in~\eqref{eq:w_hat}, which satisfies
\begin{align*}
    \hat{W}^{\star}_{t,h}(s)
    = \max_{a \in \mathcal{A}} \Big \{ r^{(\beta)}_{t,h}(s, a) + \mathbb{E}_{s' \sim \hat{P}_{t}(\cdot \mid s, a)} \big[ \hat{W}^{\star}_{t,h+1}(s') \big] \Big \},\qquad h=0, \cdots, H_{t}-1
\end{align*}
for any $s \in \mathcal{S}$, with terminal condition $\hat{W}^\star_{t,H_{t}} \equiv 0$. 

\medskip
\noindent
In addition, we denote its state--action value function by $\hat Z^{\star}_{t,h}$ as
\begin{align*}
    \hat Z^{\star}_{t,h}(s, a) = r^{(\beta)}_{t,h}(s, a) + \mathbb{E}_{s' \sim \hat{P}_{t}(\cdot \mid s, a)} \big[ \hat{W}^{\star}_{t,h+1}(s') \big],\qquad h=0, \cdots, H_{t}-1,
\end{align*}
for any $(s,a) \in \mathcal{S} \times \mathcal{A}$, where $\hat{W}^{\star}_{t,h}(s)
    = \max_{a \in \mathcal{A}} \hat Z^{\star}_{t,h}(s, a)$.

\medskip
\noindent
Now we upper bound $\Delta_{t}(s_{0})$. For that, rewrite $Q_{t}^{\star}(s_{t},a_{t}^\star) - Q_{t}^{\star}(s_{t},a_{t}^{\mathrm{alg}})$ by adding and subtracting $\hat{Z}_{t,0}^{\star}(s_{t},a_{t}^{\mathrm{alg}})$ and $\hat{Z}_{t,0}^{\star}(s_{t},a_{t}^{\star})$:
\begin{align*}
    Q_{t}^{\star}(s_{t},a_{t}^\star) - Q_{t}^{\star}(s_{t},a_{t}^{\mathrm{alg}})
    &=   Q_{t}^{\star}(s_{t},a_{t}^\star) - \hat{Z}_{t,0}^{\star}(s_{t},a_{t}^{\star})
    + \hat{Z}_{t,0}^{\star}(s_{t},a_{t}^{\mathrm{alg}}) - Q_{t}^{\star}(s_{t},a_{t}^{\mathrm{alg}})\\
    &\quad + \hat{Z}_{t,0}^{\star}(s_{t},a_{t}^{\star})
    - \hat{Z}_{t,0}^{\star}(s_{t},a_{t}^{\mathrm{alg}}).
\end{align*}

\medskip
\noindent
Since $a^{\mathrm{alg}}_{t} = \arg\max_{a} \hat{Z}_{t, 0}^{\star}(s_{t},a)$, we have
\[
\hat{Z}_{t,0}^{\star}(s_{t},a_{t}^{\star})
    - \hat{Z}_{t,0}^{\star}(s_{t},a_{t}^{\mathrm{alg}}) \leq 0,
\]
and therefore
\begin{align*}
    Q_{t}^{\star}(s_{t},a_{t}^\star) - Q_{t}^{\star}(s_{t},a_{t}^{\mathrm{alg}})
    &\leq \max_{a} \Big( Q_{t}^{\star}(s_{t},a) - \hat{Z}_{t,0}^{\star}(s_{t},a) \Big) 
    - \min_{a} \Big( Q_{t}^{\star}(s_{t},a) - \hat{Z}_{t,0}^{\star}(s_{t},a) \Big).
\end{align*}

\medskip
\noindent
Expanding the terms gives
\begin{align*}
    &Q_{t}^{\star}(s_{t},a_{t}^\star) - Q_{t}^{\star}(s_{t},a_{t}^{\mathrm{alg}})\\ 
    &\quad \leq \max_{a\in\mathcal{A}} \bigg( r_{t}(s_{t},a) + \mathbb{E}_{s'\sim P_t(\cdot\mid s_{t},a)} [V_{t+1}^{\star}(s')] - r^{(\beta)}_{t,0}(s_{t},a) - \mathbb{E}_{s'\sim \hat{P}_t(\cdot\mid s_{t},a)} [\hat{W}_{t,1}^{\star}(s')] \bigg)\\
    &\qquad - \min_{a\in\mathcal{A}} \bigg( r_{t}(s_{t},a) + \mathbb{E}_{s'\sim P_t(\cdot\mid s_{t},a)} [V_{t+1}^{\star}(s')] - r^{(\beta)}_{t,0}(s_{t},a) - \mathbb{E}_{s'\sim \hat{P}_t(\cdot\mid s_{t},a)} [\hat{W}_{t,1}^{\star}(s')] \bigg).
\end{align*}

\medskip
\noindent
Using $r_t(s,a) - r^{(\beta)}_{t,0}(s,a) = -\beta \mathrm{u}_{t|t}(s,a)$, we obtain
\begin{align*}
    Q_{t}^{\star}(s_{t},a_{t}^\star) - Q_{t}^{\star}(s_{t},a_{t}^{\mathrm{alg}})
    & = \mathrm{sp} \big( \mathrm{u}_{t \mid t} (s_{t}, \cdot) \big) 
    + \mathrm{sp} \bigg( \mathbb{E}_{s'\sim P_t(\cdot\mid s_{t},\cdot)}[V_{t+1}^{\star}(s')] 
    -  \mathbb{E}_{s'\sim \hat{P}_t(\cdot\mid s_{t},\cdot)}[\hat{W}_{t,1}^{\star}(s')] \bigg).
\end{align*}

\medskip
\noindent
Now, adding and subtracting $\mathbb{E}_{s' \sim \hat P_t (\cdot\mid s_{t},a)}[V_{t+1}^\star(s')]$ and upper bounding uniformly over $s_t$, we obtain
\begin{align*}
    Q_{t}^{\star}(s_{t},a_{t}^\star) - Q_{t}^{\star}(s_{t},a_{t}^{\mathrm{alg}})
    &\leq \max_{s} \operatorname{sp}(\mathrm{u}_{t|t}(s,\cdot)) \\
    &\quad+\mathrm{sp} \bigg( \mathbb{E}_{s'\sim P_t(\cdot\mid s_{t},\cdot)}[V_{t+1}^{\star}(s')] 
    - \mathbb{E}_{s' \sim \hat P_t (\cdot\mid s_{t},\cdot)}[V_{t+1}^\star(s')] \bigg)\\
    &\quad +\mathrm{sp} \bigg(
    \mathbb{E}_{s' \sim \hat P_t (\cdot\mid s_{t},\cdot)}[V_{t+1}^\star(s')]
    - \mathbb{E}_{s' \sim \hat P_t (\cdot\mid s_{t},\cdot)}[\hat{W}_{t,1}^{\star}(s')] \bigg).
\end{align*}

\medskip
\noindent
Using the standard bound that,
\begin{align*}
    \mathrm{sp} \bigg( \mathbb{E}_{s'\sim P_t(\cdot\mid s_{t},\cdot)}[V_{t+1}^{\star}(s')] 
        - \mathbb{E}_{s' \sim \hat P_t (\cdot\mid s_{t},\cdot)}[V_{t+1}^\star(s')] \bigg) \le
    \max_{a\in\mathcal A}\|P_t(\cdot\mid s_t,a)-\hat P_t(\cdot\mid s_t,a)\|_{\mathrm{tv}}\,
    \mathrm{sp}(V_{t+1}^\star),
\end{align*}
we get
\begin{equation}\label{eq:helper}
    \begin{aligned}
        Q_{t}^{\star}(s_{t},a_{t}^\star) - Q_{t}^{\star}(s_{t},a_{t}^{\mathrm{alg}})
        &\leq 
        \max_{s} \operatorname{sp}(\mathrm{u}_{t|t}(s,\cdot)) 
        + \max_{s,a}\| P_t(\cdot\mid s,a) - \hat{P}_t(\cdot\mid s,a) \|_{\mathrm{tv}}\ \mathrm{sp} \big(V_{t+1}^{\star}\big)\\ 
        &\quad + \mathrm{sp} \Big(
         V_{t+1}^{\star}- \hat{W}_{t,1}^{\star} \Big).
    \end{aligned}
\end{equation}

\medskip
\noindent
Now, to upper bound $\mathrm{sp} \big (V^{\star}_{t+1} - \hat{W}^{\star}_{t,1}\big )$, we apply Lemma~\ref{lem:multi_stage_error}.
Let
\begin{align*}
    \hat{\varepsilon}_{t,k} 
    &= \max_{s,a}\| P_{t+k}(\cdot\mid s,a) - \hat{P}_t(\cdot\mid s,a) \|_{\mathrm{tv}},\\
    \hat{\delta}_{t,k} &= \max_{s} \operatorname{sp}(\mathrm{u}_{t+k|t}(s,\cdot)).
\end{align*}

\medskip
\noindent
Then,
\begin{align*}
    \mathrm{sp} \big( V^{\star}_{t+1} - \hat{W}^{\star}_{t,1} \big)
    &\leq  
    \alpha^{\lfloor \frac{H_{t}-1}{m} \rfloor}  \ \mathrm{sp} \Big(  V^{\star}_{t+H_{t}} \Big)\\
    &\quad + 2 \alpha^{\lfloor \frac{H_{t}-1}{m} \rfloor} 
    \sum_{i=1}^{H_{t}-\lfloor \frac{H_{t}-1}{m} \rfloor m - 1} 
    \Big( 
    \hat{\varepsilon}_{t,\lfloor \frac{H_{t}-1}{m} \rfloor m+i}\ 
    \mathrm{sp} \big(  \hat{W}^{\star}_{t,(\lfloor \frac{H_{t}-1}{m} \rfloor )m+1} \big) 
    + \hat{\delta}_{t,\lfloor \frac{H_{t}-1}{m} \rfloor m+i}  
    \Big)\\
    &\quad + 2 \sum_{\ell = 0}^{\lfloor \frac{H_{t}-1}{m} \rfloor - 1} 
    \alpha^{\ell} 
    \sum_{i=1}^{m} 
    \Big( 
    \hat{\varepsilon}_{t,\ell m+i}\ 
    \mathrm{sp} \big(  \hat{W}^{\star}_{t,(\ell+1)m+1} \big) 
    + \hat{\delta}_{t,\ell m+i}  
    \Big),
\end{align*}
where we used that $\hat{W}^{\star}_{t,H_{t}} \equiv 0$.

\medskip
\noindent
Going back to~\eqref{eq:helper}, we conclude that  
\begin{align*}
    &Q_{t}^{\star}(s_{t},a_{t}^\star) - Q_{t}^{\star}(s_{t},a_{t}^{\mathrm{alg}})\\
    &\quad \leq 
    \hat{\delta}_{t,0} + \hat{\varepsilon}_{t,0}\ \mathrm{sp} \big(V_{t+1}^{\star}\big) 
    + \alpha^{\lfloor \frac{H_{t}-1}{m} \rfloor}  \ \mathrm{sp} \big(  V^{\star}_{t+H_{t}} \big)\\
    &\qquad + 2 \alpha^{\lfloor \frac{H_{t}-1}{m} \rfloor} 
    \sum_{i=1}^{H_{t}-\lfloor \frac{H_{t}-1}{m} \rfloor m - 1} 
    \Big( 
    \hat{\varepsilon}_{t,\lfloor \frac{H_{t}-1}{m} \rfloor m+i}\ 
    \mathrm{sp} \big(  \hat{W}^{\star}_{t,(\lfloor \frac{H_{t}-1}{m} \rfloor )m+1} \big) 
    + \hat{\delta}_{t,\lfloor \frac{H_{t}-1}{m} \rfloor m+i}  
    \Big)\\
    &\qquad + 2 \sum_{\ell = 0}^{\lfloor \frac{H_{t}-1}{m} \rfloor - 1} 
    \alpha^{\ell} 
    \sum_{i=1}^{m} 
    \Big( 
    \hat{\varepsilon}_{t,\ell m+i}\ 
    \mathrm{sp} \big(  \hat{W}^{\star}_{t,(\ell+1)m+1} \big) 
    + \hat{\delta}_{t,\ell m+i}  
    \Big).
\end{align*}

\medskip
\noindent
Now since this bound holds for any $s_t \in \mathcal{S}$, we conclude that:
\begin{align*}
    \Delta_{t}(s_{0})
    &\leq 
    \hat{\delta}_{t,0} + \hat{\varepsilon}_{t,0}\ \mathrm{sp} \big(V_{t+1}^{\star}\big) 
    + \alpha^{\lfloor \frac{H_{t}-1}{m} \rfloor}  \ \mathrm{sp} \big(  V^{\star}_{t+H_{t}} \big)\\
    &\quad + 2 \alpha^{\lfloor \frac{H_{t}-1}{m} \rfloor} 
    \sum_{i=1}^{H_{t}-\lfloor \frac{H_{t}-1}{m} \rfloor m - 1} 
    \Big( 
    \hat{\varepsilon}_{t,\lfloor \frac{H_{t}-1}{m} \rfloor m+i}\ 
    \mathrm{sp} \big(  \hat{W}^{\star}_{t,(\lfloor \frac{H_{t}-1}{m} \rfloor )m+1} \big) 
    + \hat{\delta}_{t,\lfloor \frac{H_{t}-1}{m} \rfloor m+i}  
    \Big)\\
    &\quad + 2 \sum_{\ell = 0}^{\lfloor \frac{H_{t}-1}{m} \rfloor - 1} 
    \alpha^{\ell} 
    \sum_{i=1}^{m} 
    \Big( 
    \hat{\varepsilon}_{t,\ell m+i}\ 
    \mathrm{sp} \big(  \hat{W}^{\star}_{t,(\ell+1)m+1} \big) 
    + \hat{\delta}_{t,\ell m+i}  
    \Big).
\end{align*}
\end{proof}

\section{Analysis of Regret During Skip Intervals}\label{app_2}

\begin{proof}[Proof of Lemma~\ref{lem:Q_value_diff_skip}]
Suppose $t \in \mathcal{T}_{\mathrm{skip}}$. At time $t$, Algorithm~\ref{al:main} relies on the most recent update at time $\tau_{k(t)}$. In particular, the agent uses the last updated policy evaluated at the last observed state $s_{\tau_{k(t)}}$ and executes
\begin{align}
    a_{t} \sim \pi^{\mathrm{alg}}_{\tau_{k(t)}}(s_{\tau_{k(t)}}).
\end{align}
This means that at any time $t$, the most recent policy is evaluated at the last observed state $s_{\tau_{k(t)}}$.

\medskip
\noindent
Now let $L_{t}(s_{t})$ be defined as:
\begin{align*}
    L_{t}(s_{t}) = Q^{\star}_{t} \big(s_{t}, \pi^{\star}_{t}(s_{t})\big) - Q^{\star}_{t}\big(s_{t}, \pi^{\mathrm{alg}}_{\tau_{k(t)}}(s_{\tau_{k(t)}}) \big).
\end{align*}

\medskip
We can rewrite $L_{t}(s_{t})$ as:
\begin{align*}
    L_{t}(s_t)
    &= Q^{\star}_{t} \big(s_t, \pi^{\star}_{t}(s_t)\big) - Q^{\star}_{\tau_{k(t)}} \big(s_t, \pi^{\star}_{\tau_{k(t)}}(s_t) \big)\\
    &\quad + Q^{\star}_{\tau_{k(t)}} \big(s_t, \pi^{\star}_{\tau_{k(t)}}(s_t) \big) - Q^{\star}_{\tau_{k(t)}} \big(s_t, \pi^{\mathrm{alg}}_{\tau_{k(t)}}(s_t) \big)\\
    &\quad + Q^{\star}_{\tau_{k(t)}} \big(s_t, \pi^{\mathrm{alg}}_{\tau_{k(t)}}(s_t)\big) - Q^{\star}_{t}\big(s_t, \pi^{\mathrm{alg}}_{\tau_{k(t)}}(s_t) \big)\\
    &\quad + Q^{\star}_{t} \big(s_t, \pi^{\mathrm{alg}}_{\tau_{k(t)}}(s_t)\big) - Q^{\star}_{t}\big(s_t, \pi^{\mathrm{alg}}_{\tau_{k(t)}}(s_{\tau_{k(t)}}) \big)\\
    &= \underbrace{Q^{\star}_{\tau_{k(t)}} \big(s_t, \pi^{\star}_{\tau_{k(t)}}(s_t) \big) - Q^{\star}_{\tau_{k(t)}} \big(s_t, \pi^{\mathrm{alg}}_{\tau_{k(t)}}(s_t) \big)}_{\text{(I)}} \\
    &\quad + \underbrace{\mathrm{sp} \Big( Q^{\star}_{t} (s_t, \cdot) - Q^{\star}_{\tau_{k(t)}} (s_t, \cdot ) \Big)}_{\text{(II)}} \\
    &\quad + \underbrace{Q^{\star}_{t} \big(s_t, \pi^{\mathrm{alg}}_{\tau_{k(t)}}(s_t)\big) - Q^{\star}_{t}\big(s_t, \pi^{\mathrm{alg}}_{\tau_{k(t)}}(s_{\tau_{k(t)}}) \big)}_{\text{(III)}}.
\end{align*}

\medskip
\noindent
For term $(\mathrm{I})$, we have
\begin{align*}
    \mathbb{E}_{s_{1}\sim P_{0}(\cdot \mid s_{0},a_{0})}\ \dots \mathbb{E}_{s_{t}\sim P_{t-1}(\cdot \mid s_{t-1},a_{t-1})} 
    \Big[ Q^{\star}_{\tau_{k(t)}} \big(s_t, \pi^{\star}_{\tau_{k(t)}}(s_t) \big) - Q^{\star}_{\tau_{k(t)}} \big(s_t, \pi^{\mathrm{alg}}_{\tau_{k(t)}}(s_t) \big) \Big] 
    = \Delta_{\tau_{k(t)}}(s_{0}),
\end{align*}
which can be bounded using Lemma~\ref{lem:Q_value_diff_obs}.

\medskip
\noindent
To bound term $(\mathrm{II})$, we write:
\begin{align*}
    \mathrm{sp} \Big( Q^{\star}_{t} (s, \cdot) - Q^{\star}_{\tau_{k(t)}} (s, \cdot ) \Big)
    &\leq \max_{s} \mathrm{sp} \big( r_{t}(s,\cdot) - r_{\tau_{k(t)}}(s,\cdot) \big) + \max_{s,a} \| P_{t}(\cdot\mid s,a) - P_{\tau_{k(t)}}(\cdot\mid s,a)\|_{\mathrm{tv}}\ \mathrm{sp} \big(V^{\star}_{t}\big) \\
    &\quad + \mathrm{sp} \big(V^{\star}_{t} - V^{\star}_{\tau_{k(t)}} \big).
\end{align*}

\medskip
\noindent
Let
\begin{align*}
    \bar{\varepsilon}_{\tau_{k(t)},t} &:= \max_{s,a} \| P_{t}(\cdot\mid s,a) - P_{\tau_{k(t)}}(\cdot\mid s,a)\|_{\mathrm{tv}},\\
    \bar{\delta}_{\tau_{k(t)},t} &:= \max_{s} \mathrm{sp} \big( r_{t}(s,\cdot) - r_{\tau_{k(t)}}(s,\cdot) \big),
\end{align*}
then
\begin{align*}
    \mathrm{sp} \Big( Q^{\star}_{t} (s, \cdot) - Q^{\star}_{\tau_{k(t)}} (s, \cdot ) \Big)
    \leq \bar{\delta}_{\tau_{k(t)}, t} + \bar{\varepsilon}_{\tau_{k(t)}, t} \tilde{V} + \mathrm{sp} \big(V^{\star}_{t} - V^{\star}_{\tau_{k(t)}} \big).
\end{align*}

\medskip
\noindent
To bound $\mathrm{sp} \big(\,V^{\star}_{t}\, - V^{\star}_{\tau_{k(t)}} \big)$, let us consider the auxiliary MDP, $\bar{\mathcal{M}} = \Big(\mathcal{S}, \mathcal{A}, T, \{\bar{P}_{\ell}\}_{\ell=0}^{T-1}, \{\bar{r}_{\ell}\}_{\ell=0}^{T-1} \Big)$ with reward function and probability transitions as:
\begin{align*}
    \bar r_\ell(s,a)
    &=
    \left\{
    \begin{aligned}
        & r_\ell(s,a) && \ell = 0,\dots,\tau_{k(t)}-1, \\[1.5mm]
        & r_{t+j}(s,a) && \ell = \tau_{k(t)} + j,\ j = 0,\dots,T-t-1, \\[1.5mm]
        & 0 && \ell \ge T - t + \tau_{k(t)}
    \end{aligned}
    \right. \\[4mm]
    \bar P_\ell(\cdot \mid s,a)
    &=
    \left\{
    \begin{aligned}
        & P_\ell(\cdot \mid s,a) 
        && \ell = 0,\dots,\tau_{k(t)}-1, \\[1.5mm]
        & P_{t+j}(\cdot \mid s,a) 
        && \ell = \tau_{k(t)} + j,\ j = 0,\dots,T - t-1, \\[1.5mm]
        & P_{T-1} 
        && \ell \ge T - t + \tau_{k(t)}
    \end{aligned}
    \right.
\end{align*}
\medskip
Also, let $\{\bar V_\ell^*\}_{\ell=0}^T$ denote the optimal value functions of $\bar{\mathcal M}$, i.e.
\begin{align*}
    \bar V_T^*(s) &= 0,\\
    \bar V_\ell^*(s) &= \max_{a\in\mathcal A}
    \left\{ \bar r_\ell(s,a) + \mathbb{E}_{s' \sim \bar P_\ell( . \mid s,a)}[\bar V_{\ell+1}^*(s')]
    \right\},  \quad \ell = 0,\dots,T-1.
\end{align*}
for all $s\in\mathcal S$. The following lemma relates the optimal value functions of $\bar{\mathcal{M}}$ to those of the original MDP.

\begin{lemma}[Time-shifted auxiliary MDP]
The optimal value functions of $\bar{\mathcal{M}}$,
$\{\bar V_\ell^*\}_{\ell=0}^T$ satisfy
\begin{align*}
    &\bar V_{\ell}^*(s) = V_{t - \tau_{k(t)} + \ell}^*(s),\qquad \ell = \tau_{k(t)},\dots,T-t+\tau_{k(t)},\\
    &\bar V_\ell^*(s) = 0,\qquad \qquad \ \ \ \ \ \ell = T - t + \tau_{k(t)},\dots,T, 
\end{align*}
for all $s \in \mathcal S$
\end{lemma}
\begin{proof}
Since $\bar V_T^*(s)=0$ for all $s$, backward induction yields
\begin{align*}
    \bar V_\ell^*(s) = \max_{a} \Bigl\{ 0 + \mathbb{E}_{s' \sim \bar P_\ell( . \mid s,a)}[0] \Bigr\} = 0, \qquad  \forall s,\ \forall \ell \ge T-t+\tau_{k(t)}.
\end{align*}
Thus,
\begin{align*}
    \bar V_{T-t+\tau_{k(t)}}^*(s) = 0 = V_T^*(s), \qquad \forall s\in\mathcal S.
\end{align*}
We now prove, by backward induction on $j=0, \dots, T-t$, that
\[
    \bar V_{\tau+j}^*(s) = V_{t+j}^*(s),  \qquad \forall s\in\mathcal S.
\]
For $j=T-t$ we established above:
\[
    \bar V_{T-t+\tau_{k(t)}}^*(s) = 0 =  V_T^*(s).
\]
Assume the induction hypothesis
\[
    \bar V_{\tau_{k(t)}+j}^*(s) = V_{t+j}^*(s),
    \qquad \forall s\in\mathcal S,
\]
for some $j\in\{1,\dots,T-t\}$.
Using the Bellman optimality equations of $\bar{\mathcal M}$ at time $\tau_{k(t)}+j-1$ and of the original MDP at time $t+j-1$, and the fact that
\[
    \bar r_{\tau_{k(t)}+j-1} = r_{t+j-1},
    \qquad
    \bar P_{\tau_{k(t)}+j-1} = P_{t+j-1},
\]
together with the induction hypothesis, we obtain
\begin{align*}
\bar V_{\tau+j-1}^*(s)
&=
\max_{a\in\mathcal A}
\left\{
    \bar r_{\tau_{k(t)}+j-1}(s,a)
    +
    \sum_{s'} \bar P_{\tau+j-1}(s'\mid s,a)\,
        \bar V_{\tau_{k(t)}+j}^*(s')
\right\} \\
&=
\max_{a\in\mathcal A}
\left\{
    r_{t+j-1}(s,a)
    +
    \sum_{s'} P_{t+j-1}(s'\mid s,a)\,
        V_{t+j}^*(s')
\right\} \\
&= V_{t+j-1}^*(s).
\end{align*}
Thus the claim holds for $j-1$, completing the backward induction. This completes the proof.    
\end{proof}

\noindent
Based on the above lemma we can replace $V^{\star}_{t}$ with $\bar{V}^{\star}_{\tau_{k(t)}}$ and thus we have:
\begin{align*}
    \mathrm{sp} \big(\,V^{\star}_{\tau_{k(t)}}\, - V^{\star}_{t} \big) &=  \mathrm{sp} \big(\,V^{\star}_{\tau_{k(t)}}\, - \bar{V}^{\star}_{\tau_{k(t)}} \big).
\end{align*}
Since the TVMDP $\mathcal{M}$ satisfies the Assumption~\ref{ass:mixing}, then the optimal Bellman operator contracts over blocks of length $m$. Using the results of the Lemma~\ref{lem:multi_stage_error} and considering the fact that $\bar{V}^{\star}_{\tau_{k(t)} + T - t} = V^{\star}_{T} \equiv 0$, we conclude that:
\begin{align*}
    \mathrm{sp} \big(V^{\star}_{\tau_{k(t)}} - V^{\star}_{t} \big)
    &\leq \alpha^{\lfloor \frac{T-t}{m}\rfloor}\ \tilde{V} \\
    &\quad +2 \alpha^{\lfloor \frac{T-t}{m}\rfloor}\ \sum_{i=0}^{T-t - \lfloor\frac{T-t}{m}\rfloor m } 
    \bigg( \bar{\varepsilon}_{\tau_{k(t)}, \tau_{k(t)}+\lfloor \frac{T-t}{m}\rfloor m+i}\ \tilde{V} 
    + \bar{\delta}_{\tau_{k(t)}, \tau_{k(t)}+\lfloor \frac{T-t}{m}\rfloor m+i}  \bigg)\\
    &\quad + 2 \sum_{\ell =0}^{\lfloor \frac{T-t}{m}\rfloor - 1} \alpha^{\ell }\ \sum_{i=0}^{m-1} 
    \bigg( \bar{\varepsilon}_{\tau_{k(t)}, \tau_{k(t)}+\ell m+i}\ \tilde{V} 
    + \bar{\delta}_{\tau_{k(t)}, \tau_{k(t)}+\ell m+i}  \bigg).
\end{align*}

\medskip
\noindent
For term $(\mathrm{III})$, we have
\begin{align*}
    Q^{\star}_{t} \big(s, \pi^{\mathrm{alg}}_{\tau_{k(t)}}(s)\big) - Q^{\star}_{t}&(s, \pi^{\mathrm{alg}}_{\tau_{k(t)}}(s_{\tau_{k(t)}}) \big) \leq \mathrm{sp} \big( r_{t} (s, \cdot) \big) + \max_{a,a'} \| P_{t}(\cdot | s, a) -  P_{t}(\cdot | s,a') \|_{\mathrm{tv}}\ \tilde{V}.
\end{align*}

\medskip
\noindent
Combining all bounds yields
\begin{align*}
    G_{t}(s_{0}) 
    &\leq G_{\tau_{k(t)}}(s_0) + \alpha^{\lfloor \frac{T-t}{m}\rfloor}\ \tilde{V} + \bar{\delta}_{\tau_{k(t)}, t} +   \bar{\varepsilon}_{\tau_{k(t)}, t} \tilde{V}\\
    &\quad +2 \alpha^{\lfloor \frac{T-t}{m}\rfloor}\ \sum_{i=0}^{T-t - \lfloor\frac{T-t}{m}\rfloor m } 
    \bigg( \bar{\varepsilon}_{\tau_{k(t)}, \tau_{k(t)}+\lfloor \frac{T-t}{m}\rfloor m+i}\ \tilde{V} 
    + \bar{\delta}_{\tau_{k(t)}, \tau_{k(t)}+\lfloor \frac{T-t}{m}\rfloor m+i}  \bigg)\\
    &\quad + 2 \sum_{\ell =0}^{\lfloor \frac{T-t}{m}\rfloor - 1} \alpha^{\ell }\ \sum_{i=0}^{m-1} 
    \bigg( \bar{\varepsilon}_{\tau_{k(t)}, \tau_{k(t)}+\ell m+i}\ \tilde{V} 
    + \bar{\delta}_{\tau_{k(t)}, \tau_{k(t)}+\ell m+i}  \bigg).
\end{align*}
\end{proof}

\end{document}